\documentclass[prb,aps,floatfix,amsmath,amssymb,superscriptaddress,tightenlines,twocolumn]{revtex4}
\usepackage{graphicx}
\usepackage{epstopdf}
\newcommand{\be}{\begin{equation}}
\newcommand{\ee}{\end{equation}}
\usepackage{amsmath}
\usepackage{amsfonts}
\usepackage{bm}
\usepackage{color}
\usepackage{subfigure}
\usepackage{amsthm}
\usepackage[latin1]{inputenc}

\newcommand{\bra}[1]{\left\langle #1 \right|}
\newcommand{\ket}[1]{\left|#1\right\rangle}

\newcommand{\Tr}{\textrm{Tr}}

\newtheorem{lem}{Lemma}

\newtheorem{defi}{Definition}
\newtheorem{corollary}{Corollary}

\begin{document}
\title{Perturbative analysis of topological entanglement entropy from conditional independence}
\author{Isaac H. Kim}
\affiliation{Institute of Quantum Information, California Institute of Technology, Pasadena CA 91125, USA}

\date{\today}
\begin{abstract}
We use the structure of conditionally independent states to analyze the stability of topological entanglement entropy. For the ground state of quantum double or Levin-Wen model, we obtain a bound on the first order perturbation of topological entanglement entropy in terms of its energy gap and subsystem size. The bound decreases superpolynomially with the size of the subsystem, provided the energy gap is nonzero. We also study the finite temperature stability of stabilizer models, for which we prove a stronger statement than the strong subadditivity of entropy. Using this statement and assuming i) finite correlation length ii) small conditional mutual information of certain configurations, first order perturbation effect for arbitrary local perturbation can be bounded. We discuss the technical obstacles in generalizing these results.
\end{abstract}

\maketitle
\section{Introduction}
Topological order is a new kind of order that cannot be described by Landau's symmetry breaking paradigm. Properties of these exotic phases include a ground state degeneracy that depends on the manifold, anyonic statistics, and long range entanglement.\cite{Kitaev2003,Levin2005,Kitaev2006,Levin2006} Such phases are expected to be stable against generic perturbation if its strength is sufficiently weak and its interaction range is bounded. Indeed, it was shown by several authors that the spectral stability follows under a set of reasonable assumptions.\cite{Klich2009,Bravyi2010,Michalakis2011}

If the energy gap remains open under the perturbation, one can adiabatically continue from the ground state of the original hamiltonian to the ground state of the perturbed hamiltonian.\cite{Hastings2005} Since the generator of this flow consists of quasi-local terms which decay almost exponentially, the perturbed hamiltonian has similar properties to the unperturbed hamiltonian.\cite{Hastings2005,Osborne2007,Bravyi2010} For example, one can define local operators that create defects with well-defined energies and string operators that can move around the defects. One may argue that the long range entanglement in the ground state can be preserved in a similar vein, although one must define precisely what long range entanglement is.

Long range entanglement in a 2D system refers to the nontrivial constant subcorrection term of the entanglement entropy, also known as the topological entanglement entropy.\cite{Hamma2005,Kitaev2006,Levin2006} While a  proof with full mathematical rigor has not been established to the best of author's knowledge, it is widely accepted by now that topological entanglement entropy is a universal constant that characterizes the phase of the gapped quantum many-body system. If one accepts the topological quantum field theory description of the low energy physics, there is a simple explanation as to why the topological entanglement entropy remains stable against generic perturbation.\cite{Kitaev2006} There are also mounting numerical evidences suggesting the stability.\cite{Hamma2007,Isakov2011,Jiang2012}

Presence of the long range entanglement can be interpreted as a consequence of some nontrivial nonlocal constraint. For example, in the ground state of a 2D gapped system supporting anyonic quasiparticles, total charge enclosed in some region must add up to be a trivial charge. However, the existence of the constant subcorrection term alone does not necessarily imply that the nature of the constraint is quantum. 3D toric code at finite temperature has nonlocal contributions to the entanglement entropy\cite{Castelnovo2008}, yet such state can be mapped to a Gibbs state of a classical hamiltonian under local unitary transformation.\cite{Hastings2011} We wish to understand if this nonlocal contribution to the entanglement entropy is an invariant of the phase. We would also like to understand the mechanism behind their stability, instead of arguing on the ground of effective field theory. In such pursuit, we introduce a property of these states that has apparently been unnoticed so far with few notable exceptions.

The key property is the conditional independence. Tripartite state $\rho_{ABC}$ is referred to be conditionally independent if its conditional mutual information $I(A:C|B) = S_{AB} + S_{BC} - S_{B} - S_{ABC}$ is $0$. Hastings and Poulin showed that even for a system with long range entanglement, there exists a set of subsystems that are conditionally independent.\cite{Poulin2010a} To see this, note that the entanglement entropy of a topologically ordered system can be expressed as $S_A = a|\partial A| - \gamma_A$, where $|\partial A|$ is the boundary area and $\gamma_A$ is a constant that only depends on the topology of $A$. If the topologies of $AB$, $BC$, $B$, and $ABC$ are all identical, $\gamma$ as well as the boundary contributions cancel out each other. Proving such statement for a generic quantum many-body system is a hard problem. However, the entanglement entropy formula for the ground state of some exactly solvable models are known.\cite{Hamma2005,Levin2006} For such systems, the conditional independence can be easily shown. Conditional independence is also the key idea behind quantum belief propagation(QBP) algorithm.\cite{Hastings2007c,Leifer2007} Success of the QBP algorithm indicates that the conditional mutual information for certain configuration is likely to be small for noncritical finite temperature systems.\cite{Poulin2007,Bilgin2009}

Main message of the present paper is that the conditional independence of certain subsystems strongly constrains the structure of these states so as to ensure the robustness of the nonlocal quantum correlation. Admittedly our result is limited to either i) the first order of the perturbation series of the exactly solvable models or ii) the perturbation that has a special structure. However, generalizations to both higher orders of perturbation series and more general models seem to be hindered by our lack of understanding about approximately conditionally independent states rather than that of the physical properties of such systems.

It has been known in quantum information community that the structure of states that are conditionally independent is significantly constrained compared to that of the the general quantum state.\cite{Ruskai2002,Petz2003,Hayden2004} In particular, exact conditional independence implies that subsystems form a quantum Markov chain. This structure, together with the locality of the hamiltonian, reveals why topological entanglement entropy changes very little, at least up to the first order of the perturbation series. A statement that extends to the approximate conditional independent states are not known. In fact, a classical statement that relates conditional mutual information to a Markov chain is known to be false for quantum states.\cite{Ibinson2007}

Rest of the paper is structured in the following way. In Section \ref{section:notations}, we introduce the notations. Section \ref{section:toolbox} explains the technical tools. We apply these tools in Section \ref{section:zero_temperature} and \ref{section:finite_temperature} which are respectively dedicated to the zero temperature and the finite temperature states. We discuss technical obstacles in generalizing the results to both higher orders and more general models in Section \ref{section:higher_order}. We discuss the implication of these results and some open problems in Section \ref{section:conclusion}.

\section{Notations\label{section:notations}}
The Hilbert space has a tensor product structure $\otimes_i \mathcal{H}_i$ where $\mathcal{H}_i$ corresponds to the local Hilbert space located at vertices of a square lattice. Local Hilbert space dimension is $d$. We assume periodic boundary condition with sufficiently large system size. We define a set of operators having nontrivial support on $\mathcal{H}_A$ as $\mathcal{B}(\mathcal{H}_A)$. The Boundary of subsystem $A$ is denoted as $\partial A$. $|A|$ represents the volume of $A$ and similarly $|\partial A|$ is the boundary area of $A$. We set the size of the subsystems to be $\mathcal{O}(l)$ unless specified otherwise.

We consider a family of hamiltonian $H(s)= H_0 + sV$ and study its behavior in the vicinity of $s=0$. Both the original hamiltonian $H_0= \sum_i h_i$ and the perturbation $V= \sum_i v_i$ consists of sum of terms that are supported on a ball of radius $r_0$ and the interaction strength is uniformly bounded by $J$, i.e. $\| h_i\|,\|v_i\| \leq J$. $ \| \cdots \|$ is  $l_{\infty}$ norm. We denote the spectral gap as $\Gamma(s)$.

Following Bravyi et al.'s construction\cite{Bravyi2006}, we define an approximation of a quasi-local operator as follows.
\begin{equation}
[O]_A = \frac{1}{\dim A^c}\Tr_{A^c}(O) \otimes I_{A^c}
\end{equation}
This approximation is motivated from the fact that a correlation generated by local hamiltonian falls off exponentially outside an effective lightcone. The quasilocal operators generated by such time evolution can be approximated by a local operator supported on a ball of finite radius $R$, with the correction term decreasing superpolynomially with $R$.
\begin{figure}[h!]
\includegraphics[width=2in]{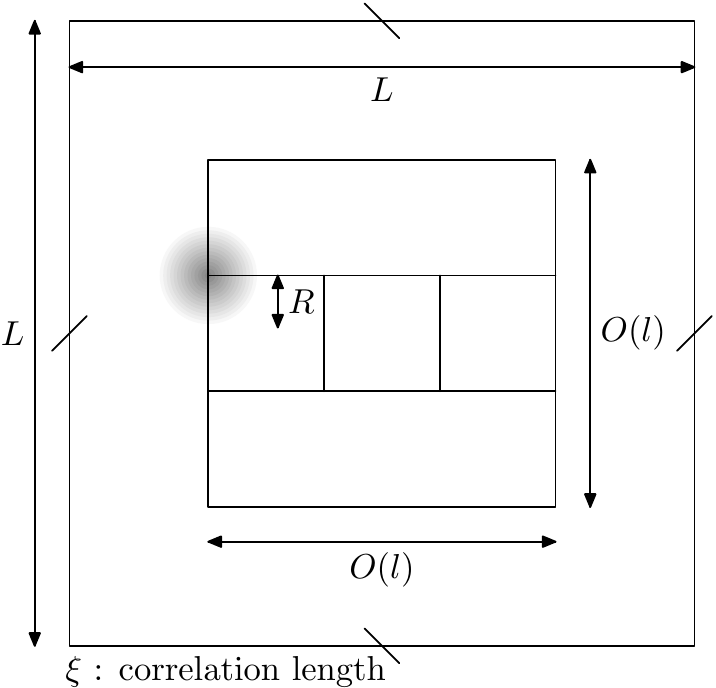}
\caption{The shaded region represents an effect of the perturbation that is smeared out in space. We shall approximate this effect by a strictly local operator with a finite radius $R$. The correction decreases superpolynomially with $R$. \label{fig:setup}}
\end{figure}

Entanglement spectrum of a subsystem $A$ is defined as $\hat{H}_A = -I_{A^c} \otimes \log \rho_A$, where $\rho_A$ is the reduced density matrix of $A$. We define conditional mutual spectrum as $\hat{H}_{A:C|B}= \hat{H}_{AB} + \hat{H}_{BC} - \hat{H}_B - \hat{H}_{ABC}$. Note that
\begin{equation}
\Tr(\rho_{ABC} \hat{H}_{A:C|B}) = I(A:C|B).
\end{equation}
We also define $\langle \cdots \rangle = \Tr(\rho \cdots)$ as an expectation value. Throughout the paper, constants $c$ and $c'$ denote numerical constants, and their exact values may be different in each contexts.

\section{Technical tools\label{section:toolbox}}
We have extensively used three technical tools in the present paper. The ideas that motivate each of these tools are as follows. First, local perturbation perturbs locally with superpolynomially decaying tail.\cite{Bachmann2011} Second, perturbation that acts locally can be bounded by using deformation moves. Deformation move refers to a chain rule of conditional mutual spectrum, analogous to the chain rule of conditional mutual information. Third, effect of the superpolynomially decaying tail can be bounded by regularizing the entanglement spectrum.

The locality estimates come from Lieb-Robinson bound.\cite{Hastings2010} The deformation moves and the regularization of the entanglement spectrum is a more refined treatment of the tools introduced in Ref.\onlinecite{Kim2012a}.

\subsection{Regularization of entanglement spectrum}
We extend some of the results in Ref.\onlinecite{Kim2012a}.
\begin{defi}
Regularized entanglement spectrum $\hat{H}_A^{\Lambda}$ with a cutoff $\Lambda$ is
\begin{equation}
\hat{H}_A^{\Lambda} = -\sum_{p \geq 1/\Lambda} \log p_i \ket{i}\bra{i},
\end{equation}
where $\ket{i}$ is an eigenstate of $\rho_A$ with an eigenvalue $p_i$.
\end{defi}
\begin{lem}
\begin{equation}
|\rho_{AB} \Delta_A^{\Lambda}|_1 \leq   \frac{d_A^3}{\Lambda^{\frac{1}{2}}} \log \Lambda ,
\end{equation}
where $\Delta_A^{\Lambda} = \hat{H}_A - \hat{H}_A^{\Lambda}$.
\end{lem}
\begin{proof}
Purify $\rho_{AB}$ to $\ket{\psi}_{ABC}$. $\ket{\psi}_{ABC}$ admits the following Schmidt decomposition.
\begin{equation}
\ket{\psi}_{ABC} = \sum_{i=1}^{d_A} \sqrt{p_i} \ket{i}_A \ket{i}_{BC},
\end{equation}
where $p_i$s are eigenvalues of $\rho_A$ and $\ket{i}_A(\ket{i}_{BC})$ are the basis states for the Hilbert space $\mathcal{H}_A(\mathcal{H}_{BC})$.

For any operator $O\in \mathcal{B}(\mathcal{H}_{AB})$, it allows the following decomposition.
\begin{equation}
O=\sum_{i=1}^{d_A^2} \sum_{i=1}^{d_{B}^2} \frac{1}{d_A d_{B}} \Tr(U_{A,i} \otimes U_{B,j} O) U_{A,i}^{\dagger} \otimes U_{B,j}^{\dagger}, \label{eq:operator_decomposition}
\end{equation}
where $U_{A,i}(U_{B,j})$ are unitary operators that are supported on $A(B)$ with appropriate normalization conditions.
\begin{align}
\Tr(U_{A,i} U_{A,j}^{\dagger}) &= d_A \delta_{ij} \nonumber \\
\Tr(U_{B,i} U_{B,j}^{\dagger}) &= d_{B} \delta_{ij}.
\end{align}
In other words, $\{U_{A,i} / \sqrt{d_A} \}$ ($\{U_{B,i} / \sqrt{d_B} \}$) is a complete set of orthonormal basis for $\mathcal{B}(\mathcal{H}_A)$ $(\mathcal{B}(\mathcal{H}_B))$ under Hilbert-Schmidt inner product $\langle O_1, O_2\rangle = \Tr(O_1^{\dagger} O_2)$. Such basis set always exists for a finite dimensional Hilbert space.\cite{Pittenger2000} Equation \ref{eq:operator_decomposition} is equivalent to the following expression.
\begin{equation}
O= \sum_{i=1}^{d_A^2} O_{B,i} \otimes U_{A,i}^{\dagger},
\end{equation}
where
\begin{align}
O_{B,i} &= \frac{1}{d_A} \Tr_A (U_{A,i} O)\\
&= \sum_{j=1}^{d_{B}} \frac{1}{d_A d_{B}}\Tr(U_{A,i} \otimes U_{B,j} O) U_{B,j}^{\dagger}.
\end{align}
Also, $O_{B,i}$ can be bounded as follows.
\begin{align}
\| O_{B,i} \| &= \frac{1}{d_A} \sup_{\ket{\phi}_{BC}} \sum_{i=1}^{d_A} \bra{\phi}_{BC} \bra{i}_A U_{A,i} O\ket{i}_A \ket{\phi}_{BC} \\
& \leq \sum_{i=1}^{d_A} \frac{1}{d_A} \|U_{A,i} O \| = \| O \|.
\end{align}
Rewriting $\Tr(\rho_{AB} \Delta_A^{\Lambda} O_{B,i} \otimes U_{A,i}^{\dagger})$ as $\bra{\psi}_{ABC} \Delta_A^{\Lambda} O_{B,i} \otimes U_{A,i}^{\dagger} \ket{\psi}_{ABC}$,
\begin{equation}
\bra{\psi}_{ABC} \Delta_{A}^{\Lambda} O_{B,i} \otimes U_{A,i}^{\dagger} \ket{\psi}_{ABC} = \Tr(\rho_A^{\frac{1}{2}} \Delta_{A}^{\Lambda} U_{A,i}^{\dagger}\rho_A^{\frac{1}{2}} \tilde{O}_{B,i}^T),\label{eq:getting_rid_of_BC}
\end{equation}
where $\tilde{O}_{B,i} = V O_{B,i} V^{\dagger}$ with isometry $V= \sum_i \ket{i}_A \bra{i}_{BC}$. $O^T$ is the transpose of $O$. Equation \ref{eq:getting_rid_of_BC} can be bounded by
\begin{equation}
|\rho_A^{\frac{1}{2}} \Delta_A^{\Lambda}|_1 \|U_{i}^{\dagger} \rho_A^{\frac{1}{2}} \tilde{O}_i^T \| \leq \frac{d_A}{\Lambda^{\frac{1}{2}}} \log \Lambda \|O_i \|.
\end{equation}
Summing over all $i$, we get
\begin{equation}
|\Tr(\rho_{AB} \Delta_A^{\Lambda} O)| \leq \| O \| \frac{d_A^3}{ \Lambda^{\frac{1}{2}}} \log \Lambda
\end{equation}
\end{proof}

Following corollaries can be easily proven by a judicious choice of $\Lambda$.
\begin{corollary}
\begin{equation}
|\Tr(\rho_{AB} \log \rho_A O)| \leq 6 \| O \| \log d_A \label{eq:correlation_entanglementspectrum}
\end{equation}
\end{corollary}

\begin{corollary}
Consider a connected correlation function $\mathcal{C}(O_1,O_2)= \langle O_1 O_2\rangle - \langle O_1\rangle \langle O_2\rangle$.
If $\mathcal{C}(O_1,O_2) \leq \|O_1 \| \|O_2 \| \epsilon$ for all $O_1, O_2$,
\begin{equation}
|\mathcal{C}(\hat{H}_A, O)| \leq \epsilon \|O \|(18 \log d_A + 4\log \frac{1}{\epsilon}).
\end{equation}
\end{corollary}

\subsection{Deformation moves}
Author has introduced a set of deformation moves to show that $\mathcal{C}(\hat{H}_{A:C|B},O)$ can be bounded for an operator $O$ supported on one of the subsystems, provided that i) area law holds approximately ii) correlation decays sufficiently fast iii) certain information-theoretic conjecture is correct.\cite{Kim2012a}

Here we construct a similar, yet slightly different statement. As in Ref.\onlinecite{Kim2012a}, the statement concerns a  correlation bound between $\hat{H}_{A:C|B}$ and an arbitrary operator $O$. The difference is that here we relax the condition on the support of $O$: $O$ is allowed to be located anywhere, as long as its support is sufficiently small compared to the subsystem.

The price we have to pay is the following. Instead of relying on an information-theoretic conjecture that is expected to hold for \emph{any} quantum states, we impose a condition on the reduced density matrices.
\begin{defi}
$\rho_{ABC}$ is $c_0$-bounded if
\begin{equation}
|\Tr_C (\rho_{ABC} \hat{H}_{A:C|B}) |_1 \leq c_0 I(A:C|B). \label{eq:c_0_bounded}
\end{equation}
\end{defi}
Note that all classical states are $1$-bounded. Reduced density matrices of finite temperature Gibbs state for the so called ``stabilizer models" are also $1$-bounded. Detailed explanation about these states shall be presented in Section \ref{section:finite_temperature}. If $I(A:C|B)=0$, conditional mutual spectrum is $1$-bounded by Petz's theorem.\cite{Petz2003} More specifically, Petz showed that
\begin{equation}
\hat{H}_{A:C|B}=0
\end{equation}
if and only if $I(A:C|B)=0$.\footnote{Here the value of the constant actually does not matter, since both sides of the inequality is $0$.}

Following Ref.\onlinecite{Kim2012a}, given a conditional mutual spectrum $\hat{H}_{A:C|B}$, we shall refer $B$ as a \emph{reference party.} $A$ and $C$ shall be referred as \emph{target parties.} Diagrammatically the reference party will be denoted with a `R' sign and the target parties will be denoted with `T' signs.

The key idea behind the deformation move is that for any local operator $O$, one can decompose $\hat{H}_{A:C|B}$ into $\hat{H}_{A_i:C_i|B_i}$ such that either i) $I(A_i:C_i|B_i)=o(1)$ or ii) $O$ is sufficiently far away from $A_iB_iC_i$. Such decomposition can be expressed as a linear combination of the following chain rule, which can be verified easily.
\begin{equation}
\hat{H}_{A_1A_2:C|B} = \hat{H}_{A_2:C|B} + \hat{H}_{A_1:C|A_2B}\label{eq:chain_rule},
\end{equation}
Nevertheless, we found it instructive to define three elementary deformation moves to explain this technique.

The first step in the deformation procedure is  to apply an \emph{isolation move}. Goal of the isolation move is to deform the boundary between the reference and the target party so that the support of $O$ is sufficiently separated from the reference party. See FIG.\ref{fig:isolation}
\begin{figure}[h!]
\includegraphics[width=3in]{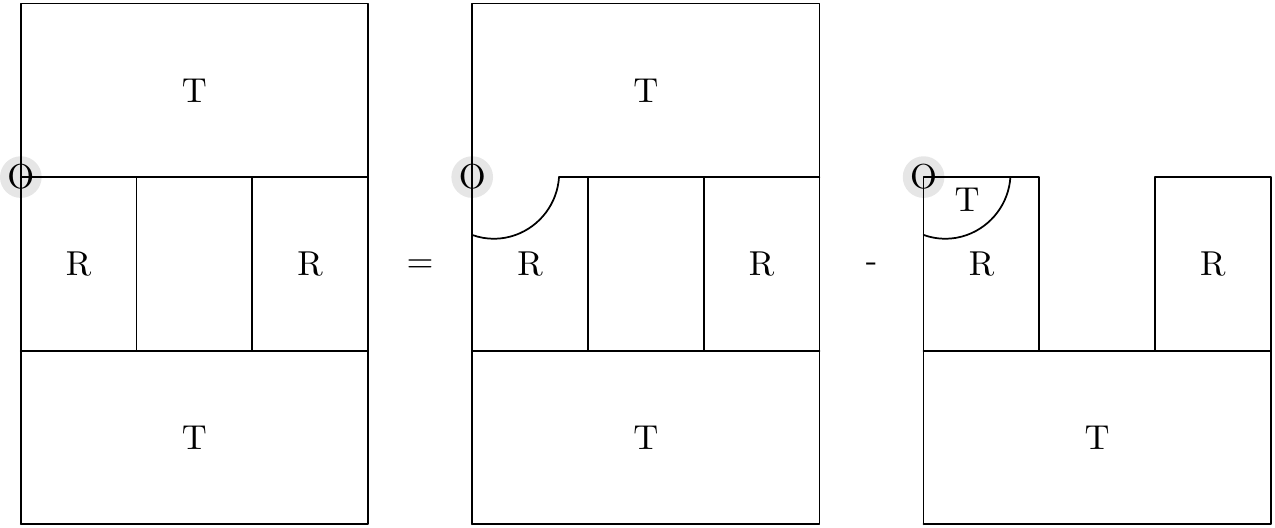}
\caption{Applying the isolation move, the conditional entanglement spectrum is deformed in such a way that i) for the new conditional entanglement spectrum, $O$ is sufficiently far away from the reference party ii) the difference is a conditional entanglement spectrum with small conditional mutual information. \label{fig:isolation}}
\end{figure}

Once the support of $O$ is isolated from the reference party, we can apply a \emph{separation move}, which separates the support of $O$ from the target parties. See FIG.\ref{fig:separation}
\begin{figure}[h!]
\includegraphics[width=3in]{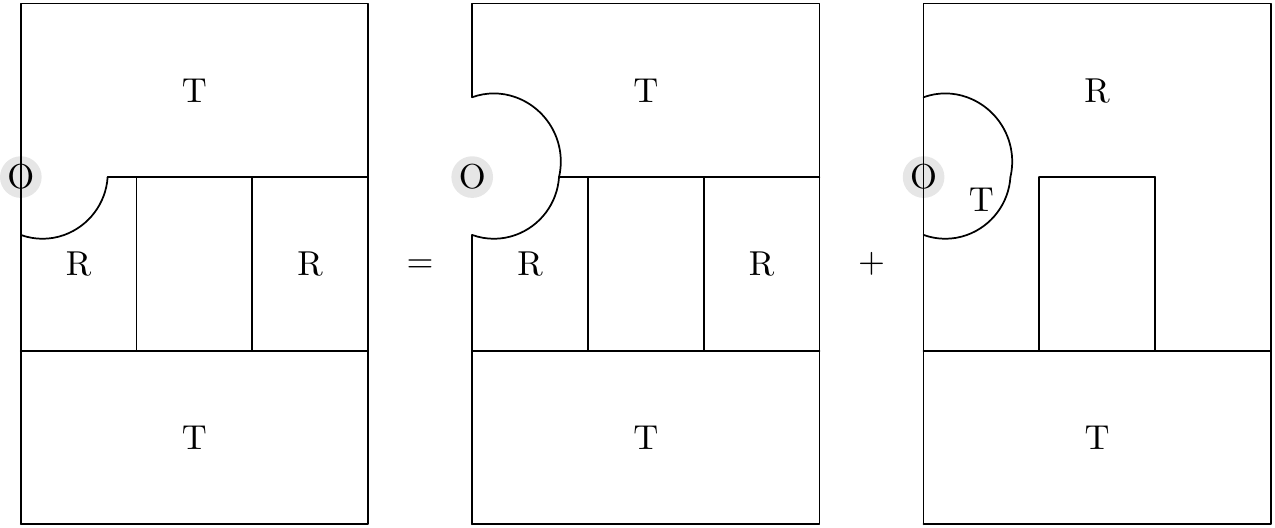}
\caption{Applying the separation move, the conditional entanglement spectrum is deformed in such a way that i) for the new conditional entanglement spectrum, $O$ is sufficiently far away from both the reference and target parties ii) the difference is a conditional entanglement spectrum with small conditional mutual information. \label{fig:separation}}
\end{figure}

Last step is to apply an \emph{absorption move}. Absorption move enables us to write the correction terms as a linear combination of $\hat{H}_{A_i:C_i|B_i}$ such that i) the support of $O$ is contained in either $A_iB_i$ or $B_iC_i$ ii) $I(A_i:C_i|B_i)=o(1)$. See FIG.\ref{fig:absorption}.
\begin{figure}[h!]
\includegraphics[width=3in]{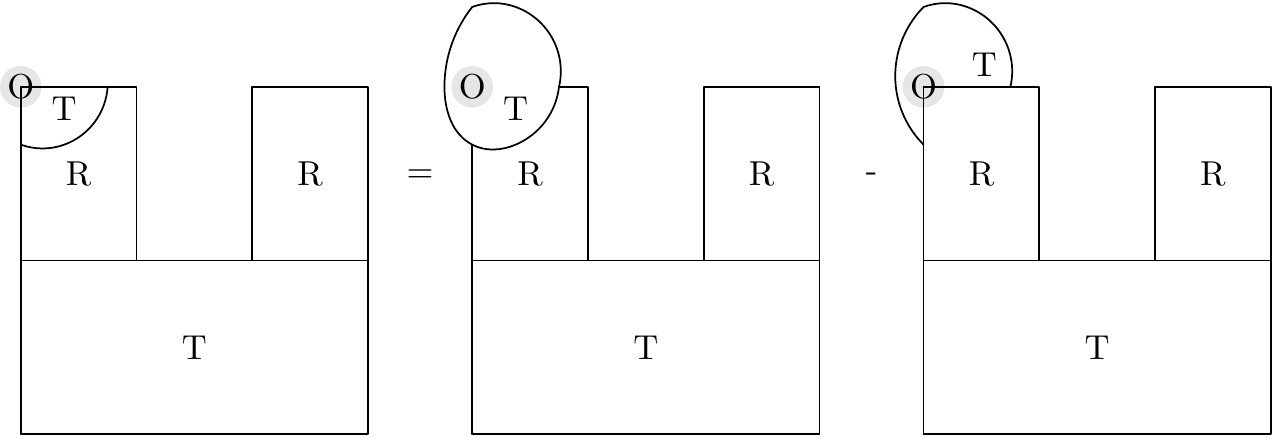}
\caption{Applying the absorption move, the conditional entanglement spectrum is expressed in terms of a linear combination of conditional entanglement spectrum $\hat{H}_{A_i:C_i|B_i}$ such that i) the support of $O$ is contained in either $A_iB_i$  or $B_iC_i$ ii) $I(A_i:C_i|B_i)$ is small. \label{fig:absorption}}
\end{figure}

To summarize, given a local operator $O$, one can decompose the conditional mutual spectrum $\hat{H}_{A:C|B}$ into $\hat{H}_{A':C'|B'}$  and correction terms with the following properties. First, the distance between $A'B'C'$ and the support of $O$ is $\mathcal{O}(l)$. Second, the correction term consists of sum of conditional mutual spectrum such that the support of $O$ is contained in the reference party and one of the target parties. Third, the conditional mutual spectra in the correction term have small conditional mutual information for the ground state of topologically ordered system.

In Section \ref{section:zero_temperature} and \ref{section:finite_temperature}, we shall frequently encounter terms of the following form.
\begin{equation}
\Tr(\rho_{ABC} \hat{H}_{A_i:C_i|B_i} O),
\end{equation}
where $O$ is an operator whose support is contained in $A_iB_i$. If $\rho_{ABC}$ is $c_0$-bounded, this term can be bounded as follows.
\begin{align}
\Tr(\rho_{ABC} \hat{H}_{A_i:C_i|B_i} O) &= \Tr_{A_i B_i} \Tr_{C_i}(\rho_{ABC} \hat{H}_{A_i:C_i|B_i} O) \nonumber \\
&\leq |\Tr_{C_i}(\rho_{ABC} \hat{H}_{A_i:C_i|B_i})|_1 \|O \| \nonumber \\
& \leq c_0I(A:C|B)\|O \|.
\end{align}

\subsection{Lieb-Robinson bound}
Lieb-Robinson bound states that there is a constant speed of light so that the correlation decays exponentially outside the effective lightcone. We refer the readers to Ref.\onlinecite{Hastings2010} for pedagogical introduction to the subject. Here we assume the quantum many-body hamiltonian satisfies the Lieb-Robinson bound and study its consequences. Given an observable $O_A(O_B)$ supported on $A(B)$, Lieb-Robinson bound can be formally stated as follows.
\begin{equation}
\| [O_A(t), O_B]\| \leq c \| O_A \| \| O_B \| \min (|A|, |B|) e^{c_1 (vt - d(A,B))},\label{eq:LR_bound}
\end{equation}
where $0<c, c_1,v < \infty$ are some constants that depend on the parameter of the hamiltonian and $d(A,B)$ is a distance between $A$ and $B$. $O(t)=e^{-iHt} O e^{iHt}$ is a time evolution of operator $O$ under the hamiltonian.

Similar, albeit weaker locality bound holds for $\int^{\infty}_{-\infty} f(t) O_A(t) dt$ if $f(t)$ decays sufficiently fast. To state this more formally, we introduce a superoperator defined as follows.
\begin{equation}
\Phi_f(O) = \int^{\infty}_{-\infty} e^{-iHt} O e^{iHt} f(t) dt
\end{equation}
 It is worth noting that in the energy eigenbasis,
\begin{equation}
\Phi_f(O)|_{ij} = \tilde{f}(E_i - E_j) O_{ij},
\end{equation}
where $\tilde{f}(\omega)$ is an inverse Fourier transform of $f(t)$.

We also define a truncated superoperator $\Phi_{f}^T$ by introducing a cutoff $T$.
\begin{equation}
\Phi_{f}^T(O) = \int^{T}_{-T} e^{-iHt}Oe^{iHt} f(t) dt.
\end{equation}

Lieb-Robinson type locality bound for $\Phi_{f}$ can be established as follows.
\begin{equation}
\|[\Phi_f(O_A), O_B] \| \leq \| [\Phi_f^T(O_A), O_B]\| + \|[\Delta_f^T(O_A), O_B] \|,
\end{equation}
where $\Delta_f^T = \Phi_f - \Phi_f^T$. The first term can be bounded by
\begin{equation}
\int^{T}_{-T} |f(t)| dt \|O_A \|\|O_B \| ce^{c' (vT-d(A,B))}
\end{equation}
from the Lieb-Robinson bound. Second term can be bounded by
\begin{equation}
\int_{\mathbb{R} \setminus [-T,T]} |f(t)| dt \|O_A \| \| O_B\|.
\end{equation}

Depending on the function $f$, one can optimize the bound with a judicious choice of $T$. An example that illustrates this idea is $\tilde{f}_1^{\beta}(\omega) = \frac{\tanh(\beta \omega / 2)}{ \beta \omega / 2}.$
\begin{lem}
If $H$ satisfies Lieb-Robinson bound,
\begin{equation}
\| [\Phi_{f_1^{\beta}}(O_A), O_B]\|  \leq c \|O_A \| \|O_B \| \min(|A|, |B|)e^{-\frac{c'd(A,B)}{ 1+ c'v\beta / \pi }},
\end{equation}
for some constant $ 0<c,c' < \infty$.
\end{lem}
Therefore, $\Phi_{f_1^{\beta}}(v_i)$ can be approximated by a local operator.
\begin{corollary}
\begin{equation}
\|\Phi_{f_1^{\beta}}(v_i) - [\Phi_{f_1^{\beta}}(v_i)]_{v_i(r)}  \| \leq c' \| v_i\| e^{-\frac{c'r}{1+c'v\beta / \pi}},
\end{equation}\label{corollary:local_approximation}
\end{corollary}
where $v_i(r)$ is a set of sites whose distance from the support of $v_i$ is less or equal to $r$.

It would be remiss if we do not mention $\Phi_{f_1^{\beta}}$ was originally introduced by Hastings in the context of QBP algorithm and finite temperature correlation decay properties of a fermionic system.\cite{Hastings2004,Hastings2007c} In fact, it is not a coincidence that these operators appear in seemingly different contexts. As one can see in the following lemma, $\Phi_{f_1^{\beta}}$ is a quantum channel that appears naturally when computing a directional derivative of a density matrix.\footnote{To see that $\Phi_{f_1^{\beta}}$ is a quantum channel, note that it has an integral Kraus representation. Furthermore, one can easily check from the normalization of $f_1^{\beta}$ that this channel is trace preserving.}

\begin{lem}
For $\rho(s)= \frac{e^{-\beta H(s)}}{Z}$,
\begin{equation}
\frac{d}{ds}\rho(s)|_{s=0} = \frac{\beta}{2}( \Phi_{f_1^{\beta}}(V) \rho_s +h.c.) - \beta \langle \Phi_{f_1^{\beta}}(V) \rangle,
\end{equation} \label{lemma:first_order_perturbation}
where $h.c.$ is hermitian conjugate.
\end{lem}

Similar technique was used by Bachman et al.\cite{Bachmann2011} in the studies of the ground state of gapped systems. They showed that a unitary evolution that adiabatically connects the ground states of two different hamiltonian is generated by a path dependent generator of the following form.
\begin{equation}
-i\frac{d}{ds}U(s) = D(s) U(s), U(0)=I.
\end{equation}
\begin{equation}
D(s) = \Phi_{W_{\Gamma}}( \frac{dH(s)}{ds})
\end{equation}
where $\Gamma= \min_{s\in [0,1]} \Gamma(s)$ and $W_{\Gamma}(t)$ is some superpolynomially decaying function. In our setting, $\frac{dH(s)}{ds}=V$. Each of the local terms $v_i$ in $V$ can be approximated as follows.\cite{Bachmann2011}
\begin{equation}
\| \Phi_{W_{\Gamma}}(v_i) - [\Phi_{W_{\Gamma}}(v_i)]_{v_i(r)} \| \leq C \| v_i \| G^{(I)}(\frac{\Gamma r}{2v}),\label{eq:QAC_approximation}
\end{equation}
where $v$ is the Lieb-Robinson velocity appearing in Equation \ref{eq:LR_bound}, and $G^{I}(x)$ is a function that satisfies the following property.
\begin{align}
G^{(I)}(x) &= \frac{K}{\Gamma} & 0\leq x \leq x_0 \nonumber \\
&= 130e^2 x^{10} u_{2/7}(x) & x> x_0.
\end{align}
Estimates for the constants are $K \approx 14708$, $36057<x_0 <36058$.\cite{Bachmann2011} Also, $u_a(x)$ is defined as follows.
\begin{equation}
u_a(x) = e^{-a \frac{x}{\ln^2 x}}.
\end{equation}

\section{Ground state of exactly solvable models\label{section:zero_temperature}}
Exact formula for the entanglement entropy is known for quantum double and Levin-Wen models.\cite{Hamma2005,Levin2006,Flammia2009} If the subsystem is simply connected, the entanglement entropy satisfies area law.
\begin{equation}
S_A = a|\partial A| - \gamma,
\end{equation}
where $\gamma$ is the topological entanglement entropy. These systems have zero correlation length, so the density matrices of two nonoverlapping regions factorize, i.e. $\rho_{AB} = \rho_A \otimes \rho_{B}$. Therefore, the following formula holds for the entanglement entropy.
\begin{equation}
S_{AB} = S_A + S_B
\end{equation}
if $A \cap B = \emptyset$.

Using standard perturbation theory, for a family of quantum states $\rho(s)$ that are differentiable with respect to $s$,
\begin{equation}
\frac{dS_A}{ds} = \Tr(\frac{d\rho}{ds} \hat{H}_A).
\end{equation}
Therefore,
\begin{align}
\frac{dI(A:C|B)}{ds} &= \Tr(\frac{d\rho}{ds} \hat{H}_{A:C|B}) \nonumber \\
 &= i \sum_j \Tr([\Phi_{W_{\Gamma}}(v_j), P_0]\hat{H}_{A:C|B}),
\end{align}
where $P_0$ is a projector onto the ground state.

Without loss of generality, let us consider terms $v_j$ that are distance $al$ or less away from $ABC$, where $a>0$ is some constant. Using deformation moves, $\hat{H}_{A:C|B} = \hat{H}_{A':C'|B'} + \sum_i a_i \hat{H}_{A_i:C_i|B_i}$, where $d(v_j,A'B'C')=\mathcal{O}(l)$ and $I(A_i:C_i|B_i)=0$. By Petz's theorem, $\hat{H}_{A_i:C_i|B_i}=0$. Now approximate $\Phi_{W_{\Gamma}}(v_j)$ by $[\Phi_{W_{\Gamma}}(v_j)]_{v_j(cl)}$ for some $c>0$ such that the support of $[\Phi_{W_{\Gamma}}(v_j)]_{v_j(cl)}$ does not overlap with $A'B'C'$. This implies the following relation.
\begin{equation}
\Tr([[\Phi_{W_{\Gamma}}(v_j)]_{v_j(cl)}, P_0] \hat{H}_{A':C'|B'})=0.
\end{equation}

To see this, consider an operator $O$ that is supported on one of $A',B',C',$ or $D=(A'B'C')^c$.
\begin{align}
i\Tr([O, P_0] \hat{H}_{A':C'|B'}) &= \frac{d}{dt}\Tr(e^{iOt} P_0 e^{-iOt} \hat{H}_{A':C'|B'}) \nonumber \\
&= \frac{d}{dt}I(A':C'|B'),
\end{align}
where the infinitesimal generator generates a unitary transformation supported on $(A'B'C')^c$. Since the entanglement entropy is invariant under local unitary transformation, this is $0$. The correction terms are of the following form.
\begin{equation}
i\Tr([\Phi_{W_{\Gamma}}(v_j)-[\Phi_{W_{\Gamma}}(v_j)]_{v_j(cl)},P_0] \hat{H}_{A':C'|B'}).
\end{equation}
Using Equation \ref{eq:correlation_entanglementspectrum} and \ref{eq:QAC_approximation}, we conclude that the effect of each terms are bounded by $cJ G^{(I)}(c'\frac{\Gamma l}{2v})l^2d$ for some constant $c$ and $c'$.

Terms that are distance $al$ or more away from $ABC$ can be bounded by approximating $\Phi_{W_{\Gamma}}(v_j)$ as $[\Phi_{W_{\Gamma}}(v_j)]_{v_j(R)}$, where $R$ is the distance between $v_j$ and $ABC$. Combining all of these contributions together, we arrive at the following bound.
\begin{equation}
\frac{d\gamma}{ds}|_{s=0} \leq cJ (\frac{\Gamma l}{v})^{10} l^4 u_{2/7}(c'\frac{\Gamma l}{v})\label{eq:bound_groundstate}
\end{equation}
for a sufficiently large subsystem size $l$. One can see that the bound diverges for gapless systems.

We note in passing that the same technique can be applied to topologically trivial configurations, i.e. $I(A:C|B)=0$. Under general perturbation that consists of sum of short-range bounded-norm terms, conditionally independent configurations become approximately conditionally independent. One may wish to establish a bootstrapping argument that recursively uses the approximate conditional independence of these configurations. Main difficulty of this approach lies on proving the $c_0$-boundedness.

\section{Stabilizer models at finite temperature\label{section:finite_temperature}}
Unlike the ground state of the exactly solvable models, exact formula for the entanglement entropy of a finite temperature system is not known except for few special cases.\cite{Castelnovo2007,Castelnovo2008,Iblisdir2010} To cope with this difficulty, we make a nontrivial but natural assumption: that the corrections from the deformation moves consist of conditional mutual spectrum with small conditional mutual information. For 3D toric code, topological entanglement entropy does not depend on the size of the subsystem for a sufficiently large subsystem.\cite{Castelnovo2008}  We shall denote the conditional mutual information in the correction terms as $\epsilon(l)$ and study how the first order perturbation effect depends on it.\footnote{Something that one must be careful about is the invariance of the topological entanglement entropy under \emph{arbitrary} deformation. Castelnovo and Chamon proved size independence in Ref.\onlinecite{Castelnovo2008}, but that does not necessarily imply invariance under arbitrary small deformation. In this paper, we have implicitly assumed the invariance under arbitrary deformation.} We shall also assume that the correlation decays exponentially.
\begin{equation}
\mathcal{C}(O_A, O_B) \leq \| O_A\| \| O_B \| \min(|A|, |B|) e^{-d(A,B)/\xi}.
\end{equation}

Stabilizer model refers to a hamiltonian of the following form
\begin{equation}
H= -\sum_i J_i s_i,
\end{equation}
where $J_i>0$ are coupling constants and $s_i$s are elements of the stabilizer group. Stabilizer group is an abelian subgroup of Pauli group. \cite{Nielsen2000} Important examples include toric code, color code, their higher dimensional generalizations, and other glassy topologically ordered systems in 3D.\cite{Kitaev2003,Bombin2006,Castelnovo2008,Chamon2005,Kim2011,Haah2011} Important property of the stabilizer models is that their reduced density matrices commute with each other.
\begin{lem}
$\rho_A = \sum_{S_i \in S(A)}c_i S_i$ for some coefficients $\{ c_i\}$.
\end{lem}
\begin{proof}
$\rho$ can be expanded as a sum of stabilizer group elements. After taking the partial trace, any operator that has nontrivial support on $A^c$ vanishes. Any stabilizer group element that has nontrivial support only on $A$ survives. But these terms are generated from the generator of the stabilizer group, so they are again elements of the stabilizer group.
\end{proof}
It trivially follows that for the Gibbs state of the stabilizer hamiltonian, reduced density matrices commute with each other. Therefore, any reduced density matrix $\rho_{ABC}$ for the stabilizer model is $1$-bounded. To see this, note the following inequality
\begin{equation}
D_1(\ln D_1 - \ln D_2) \geq D_1 - D_2
\end{equation}
for positive semidefinite operators $D_1,D_2$ which commute with each other. Setting $D_1= \rho_{ABC}$ and $D_2= \rho_{AB} \rho_{B}^{-1} \rho_{BC}$ and taking a partial trace over $C$, we conclude that $\Tr_{C}(\rho_{ABC} \hat{H}_{A:C|B})$ is a positive semidefinite operator. Since $l_1$ norm is equal to the trace for positive semidefinite operator, $\rho_{ABC}$ is $1$-bounded.

Consider terms $v_j$ that are distance $al$ or less away from $ABC$. Using the deformation moves, $\hat{H}_{A:C|B}=\hat{H}_{A':C'|B'} + \sum_i a_i \hat{H}_{A_i:C_i|B_i}$, where $d(v_j,A'B'C')=\mathcal{O}(l)$ and $I(A_i:C_i|B_i)=\epsilon(l)$. Choose an approximation radius $R$ such that $\Phi_{f_1^{\beta}}(v_j)$ is approximated by $[\Phi_{f_1^{\beta}}(v_j)]_{v_j(R)}$. First order effect of $v_j$ can be divided into three parts: the connected correlation between $[\Phi_{f_1^{\beta}}(v_j)]_{v_j(R)}$ and $\hat{H}_{A':C'|B'}$, the approximation error of $\Phi_{f_1^{\beta}}(v_j)$, and the corrections from the deformation moves. Terms that are distance $al$ or more away from $ABC$ can be similarly bounded by using the exponential correlation decay and making a judicious choice for the approximation radius $R$. All of these effects combined together results in the following bound.
\begin{equation}
\frac{1}{\beta J}\frac{d\gamma}{ds}|_{s=0} \leq \mathcal{O}(l^{2D}(e^{-c_1l/\xi}) + \mathcal{O}(l^{2D}e^{-c_2l/\beta})) +\mathcal{O}(l^{D}\epsilon(l)),
\end{equation}
where $D$ is the number of spatial dimensions and $c_1,c_2,c_3$ are some numerical constants.

\section{Comment on higher order terms\label{section:higher_order}}
Close inspection of the first order bound reveals that the $c_0$-boundedness plays a pivotal role in the derivation. For example, consider a perturbed ground state of the topologically ordered system which satisfies area law approximately. Equation \ref{eq:bound_groundstate} is only modified by including the area law correction terms, provided $c_0$-boundedness is guaranteed.

It turns out that the $c_0$-boundedness in a finite neighborhood of $s$ implies a nontrivial bound for the higher order terms as well. The key idea is that  Equation \ref{eq:bound_groundstate} can be applied to topologically trivial configuration as well as topologically nontrivial configuration. Since Equation \ref{eq:bound_groundstate} relied on the fact that the conditional mutual information of topologically trivial configuration is small, we can bootstrap this argument to the bound higher order terms.

Assuming the $c_0$-boundedness for $s \in [0,s_0)$, following inequality holds.
\begin{equation}
|\frac{d}{ds}I(A:C|B)_s| \leq \delta_s(l) + \sum_i a_i I(A_i:C_i|B_i)_s,
\end{equation}
where $\delta_s(l)$ is a function that decreases superpolynomially with $l$, and $a_i$ is a finite number that is uniformly bounded for $s \in [0,s_0]$. $I(A_i:C_i|B_i)_s$ is a conditional mutual information appearing in the correction terms of the deformation moves.

If the energy gap remains open for $s \in [0, s_0)$, $\delta_s(l)$ can be uniformly bounded by some $\delta(l)$ that decays superpolynomially in $l$. As a result, one can obtain the following recursive bound.
\begin{align}
|\gamma_s - \gamma_0| &\leq \int^{s}_{0} \delta(l) + \sum_i a_i I(A_i:C_i|B_i)_{s'} ds' \nonumber \\
&= s\delta(l) + \sum_i a_i \int^{s}_{0} \int^{s'}_{0} \frac{dI(A_i:C_i|B_i)_{s''}}{ds''}ds'' ds',
\end{align}
Here we used the fact that the conditional mutual information arising from the deformation move is $0$  at $s=0$. Recursively applying this logic, the second order term can be bounded by $\mathcal{O}(l^2 \delta(l))$. Higher order terms can be obtained in a similar manner.

To investigate the validity of the $c_0$-boundedness for general quantum many-body system, we have generated random density matrices and studied a relationship between both sides of Equation \ref{eq:c_0_bounded}. The result is plotted in FIG.\ref{fig:c_0_boundedness}. For pure states, we have applied a random unitary from Haar measure. For mixed states, we have randomly generated eigenvalues from uniform distribution over $[0,1]$, normalized, and applied random unitary from Haar measure. It seems that for certain states that have small conditional mutual information, smallest value of $c_0$ increases significantly. For this reason, we urge the readers to be careful in using this condition in general. This difficulty can be circumvented for stabilizer models against stabilizer perturbations, since commutativity of the reduced density matrices is preserved. However, it remains to be seen if the correction terms from the higher order deformation moves are small.
\begin{figure}[h!]
\includegraphics[width=0.42\textwidth]{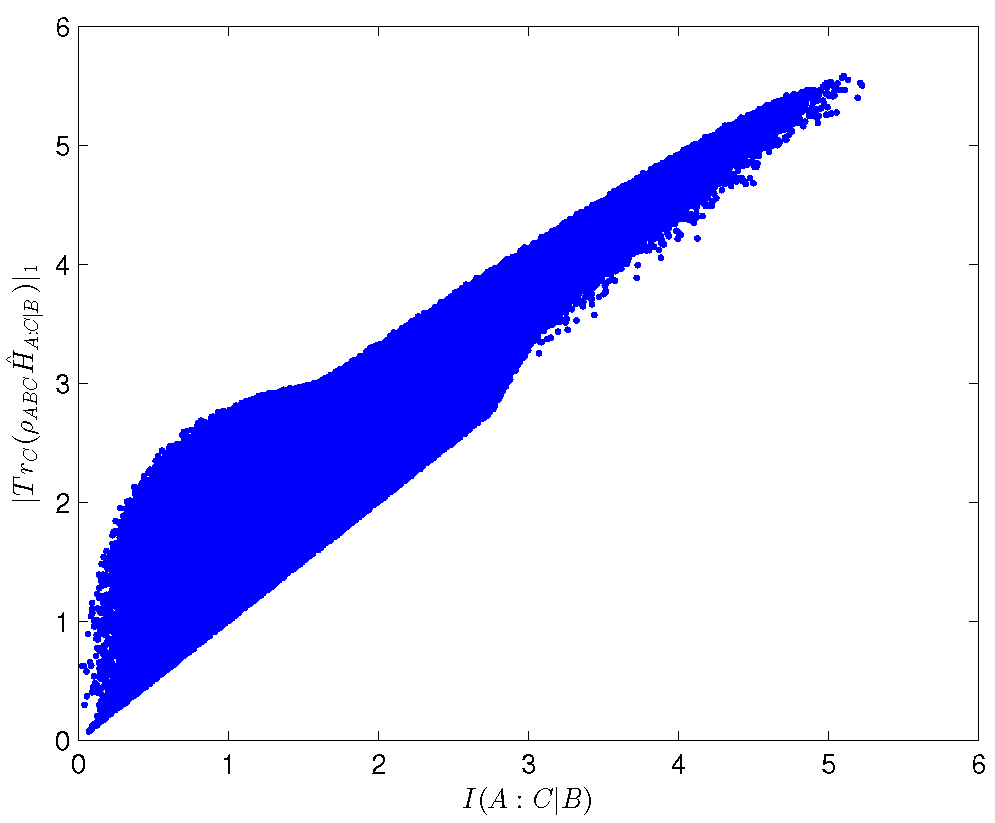}
\caption{We have numerically computed $I(A:C|B)$ and $|\Tr_C \rho_{ABC} \hat{H}_{A:C|B}|_1$ for $10^6$ randomly generated pure states. Largest observed ratio $|\Tr_C \rho_{ABC} \hat{H}_{A:C|B}|_1 / I(A:C|B)$ was $24.21924$. \label{fig:c_0_boundedness}}
\end{figure}

\section{Conclusion\label{section:conclusion}}
We have demonstrated that conditional independence strongly constrains the structure of quantum many-body system so as to ensure the first order perturbative stability of the topological entanglement entropy. Admittedly, our technique gives bounds in limited settings where i) exact conditional independence is achieved or ii) reduced density matrices commute with each other. However, once these conditions are met, the argument can be applied quite generally. In particular, we expect our method to be applicable to the studies of Chamon's model and Haah's model.\cite{Chamon2005,Haah2011} These models satisfy the topological quantum order conditions introduced by Bravyi et al, and their hamiltonian consists of sum of frustration-free commuting projectors.\cite{Bravyi2010} Therefore, the energy gap is protected against generic perturbation that consists of sum of short-range bounded-norm terms.

There are compelling reasons to believe that these models are not described by BF theory or multiple stacks of Chern-Simons theory: movement of the quasiparticles are constrained in a peculiar manner, and their ground state degeneracy is determined by some number-theoretic function that depends on the size of the system.\cite{Bravyi2010b,Bravyi2011} It would be  interesting if one can apply our method to find a linear combination of entanglement entropy that allows the first order perturbative stability.

We have also shown that our method can be extended to higher orders of perturbation series if the $c_0$-boundedness holds in a finite neighborhood of $s$, but such statement seems unlikely to hold for general quantum states. It would be very interesting if one can find an alternative technique that relies on the conjecture introduced in Ref.\onlinecite{Kim2012a}. There author was able to show that the connected correlation function between conditional mutual spectrum and local operator vanishes if the local operator is supported on one of the subsystems, provided certain extension of strong subadditivity is true for general quantum states. Unfortunately, local operators that are supported on multiple subsystems are bound to appear, as shown in the analysis of the present paper.

As for the finite temperature topological entanglement entropy in 3D, we needed two nontrivial assumptions to bound the first order perturbation effect. First, the connected correlation function between two observables decay exponentially. Second, the correction terms from the deformation moves can be expressed as a sum of small conditional mutual information. We emphasize that neither of these assumptions were explicitly proved. Further studies in explicitly bounding both of these terms are necessary.

While the structure of conditionally independent state is relatively well understood, much less is known about the states that are approximately conditionally independent. We hope our work motivates further studies in such direction.

{\it Acknowledgements---} This research was supported in part by NSF under Grant No. PHY-0803371, by ARO Grant No. W911NF-
09-1-0442, and DOE Grant No. DE-FG03-92-ER40701. Author would like to thank Spyridon Michalakis, Steve Flammia, Jeongwan Haah, Sergio Boixo, Alioscia Hamma for helpful discussions. Author would also like to thank the referees for their helpful suggestions in revising the paper.

\bibliography{bib}
\appendix

\end{document}